\documentclass{article}
\usepackage{amsmath,amssymb,bm,amsthm}
\usepackage{graphicx,epsfig,float}
\usepackage[toc,page]{appendix}
\usepackage{caption}
\usepackage{authblk}

\newcommand{\be}{\begin{equation}}
\newcommand{\ee}{\end{equation}}

\newtheorem{theorem}{Theorem}

\newtheorem{definition}[theorem]{Definition}

\title{
Mitigating an epidemic on a geographic network using vaccination
}

\author[1]{\normalsize{Mohamad Badaoui }\thanks{mohamad.badaoui1@gmail.com}}
\author[1]{\normalsize{Jean-Guy Caputo }\thanks{caputo@insa-rouen.fr}}
\author[2]{\normalsize{Gustavo Cruz-Pacheco}\thanks{cruz@mym.iimas.unam.mx}}
\author[1]{\normalsize{Arnaud Knippel }\thanks{arnaud.knippel@insa-rouen.fr}}

\affil[1]{ Laboratoire de Math\'ematiques, INSA Rouen Normandie, Av. de l'universit\'e, 76801 Saint-Etienne du Rouvray, France.}
\affil[2]{ Depto. Matem\'{a}ticas y Mec\'{a}nica,
I.I.M.A.S.-U.N.A.M., Apdo. Postal 20--726, 01000 M\'{e}xico D.F., M\'{e}xico.}

\begin{document}
\maketitle

\date{\today} 
 
\begin{abstract}
We consider a mathematical model describing the propagation of an epidemic 
on a geographical network. The size of the outbreak is
governed by the initial growth rate of the disease 
given by the maximal eigenvalue of the epidemic matrix formed by the susceptibles 
and the graph Laplacian representing the mobility. We use matrix 
perturbation theory to analyze the epidemic matrix and define
a vaccination strategy, assuming the vaccination reduces the susceptibles. \\
When mobility and local
disease dynamics have similar time scales, it is most efficient to vaccinate
the whole network because the disease grows uniformly.
However, if only a few vertices can be
vaccinated then which ones do we choose? We answer this question, 
and show that it is most efficient to vaccinate along an eigenvector 
corresponding to the largest eigenvalue of the Laplacian. 
We illustrate these general results on a 7 vertex graph and a realistic 
example of the french rail network. \\
When mobility is slower than local disease dynamics, the epidemic grows on
the vertex with largest susceptibles. 
The epidemic growth rate is more
reduced when vaccinating a larger degree vertex; it also depends on the
neighboring vertices.
This study and its conclusions provides guidelines for the planning
of vaccination on a network at the onset of an epidemic.
\end{abstract}

{\bf keywords :}  SIR epidemic model,  Graph,  Matrix perturbation  



%
%

 \section{Introduction}

The previous COVID19 epidemic
confirmed that mobility between countries or within a country is
crucial to transmit diseases. 
A set of cities or countries can be described
as vertices of a graph where edges represent communication links between
them. A first coarse-grained approach based on complex networks 
(see for example the book \cite{complex}) assumes each vertex can
have two states: healthy or infected and that a transition matrix gives 
the probability for a vertex to infect it's neighbor. This model
can describe for example the propagation of a computer virus or a rumor on
the internet \cite{holme02}. An important notion here is centrality, i.e.
the number of links associated to each vertex; vertices
with large degrees play an important role in the propagation. 
The advantage of such a model is that the network is considered as a 
whole and one can rapidly estimate how infected it is; the disease 
dynamics is however crudely represented.

To better describe the disease dynamics, one can assume 
that each vertex has a population of Susceptible and Infected individuals. 
A pioneering study was conducted by Brockmann and Helbling \cite{bh13}
to analyze the propagation of influenza via airline routes.
The mobility was described by an origin-destination probability matrix.
To introduce more details in the mobility, a number of authors
use metapopulations.
This consists, for each city $i$, in counting individuals who stay at $i$ 
and others who travel to another city $j$, see for example the nice discussions 
by Keeling et al \cite{keeling2010} and Sattenspiel and Dietz 
\cite{satdietz95}. Later Colizza et al \cite{colizza08} analyzed
the model in detail and introduced the concepts of local and global epidemic 
tresholds. Poletto et al \cite{poletto2013} also examined how fluctuations 
of the mobility fluxes affect these thresholds. From another point of view, 
Gautreau et al \cite{gautreau08} used a similar model and statistical 
physics methods to predict the arrival of a disease in a country. 
Gao \cite{gao19} studied a simpler model where populations are split
into frequent and rare travelers; he analyzed a two patch system and found
that in general diffusion reduces disease spread. See also the
analysis of Cantin and Silva \cite{cantin} and their results on a two patch 
network.
All these models are very beautiful conceptually, however their analysis 
is complicated even for moderately sized networks and it is 
not easy to have a picture of the global state of the network.
Also, it is not always possible to obtain and predict population movements.
{Finally, these models have many parameters and these are 
not easy to estimate from real data.}

In a recent article, we considered an SIR model on each
vertex where the vertices are coupled by a graph Laplacian \cite{bcc21}.
This is a symmetric version of the mobility matrix of \cite{bh13}.
In our approach, we considered that the diffusion modeled by the
graph Laplacian was small. We could estimate the diffusion
coefficient by examining the arrival times of the epidemic front in
China, Vietnam, Iran, Italy, etc.. and using this, could
predict the arrival of the COVID-19 epidemic
in Mexico \cite{cbc20}. We could also analyze 
how deconfining a large city connected to a smaller region could cause
a secondary outburst in the smaller city.

The problem of vaccine allocation introduces another complication
because the disease dynamics is in general nonlinear. There is a large literature
on the strategies for vaccination, most studies focus on preventing deaths
and hospitalizations. 
The vaccination can also be used to prevent geographical dissemination 
of the disease. For example, Matrajt et al \cite{matrajt2013} studied 
vaccine allocations at the onset of an epidemic. They coupled 
a mathematical model with a genetic algorithm to optimally distribute 
vaccine in a complete graph of asian cities and found 
that it is best to distribute vaccines and that
an epidemic can be mitigated if the vaccination campaign is fast and occurs
within the first few weeks.
Optimal control can also be used to allocate vaccines, like in the 
article by Lemaitre et al
\cite{gatto22} who choose the total number of cases as an objective function 
to be minimized. For the different scenarii they consider, the authors
find that a global strategy at network level is more effective.

We also consider the problem of a number of vaccine doses to be distributed on
the network at the onset of an epidemic. We assume
that vaccination prevents the dissemination of the disease.
This is a first approximation because it is known that some vaccines 
prevent transmission efficiently (chickenpox for example) while others 
do not (COVID-19) , depending on the infectious agent. 
The reduction of the susceptibles is small because a small proportion 
of the population is usually vaccinated. 
Matrajt introduced an epidemic prevention potential to measure the effect
of vaccination. In a similar way, in our study \cite{bcc21} we 
defined the epidemic growth rate as the maximum eigenvalue $\lambda$ of 
the epidemic matrix $M$ sum of the diagonal matrix ${\rm diag}( \beta S -\gamma) $and the graph Laplacian mobility matrix. 
If $\lambda$ is large, the maximum number of infected will be large 
and vice-versa so that $\lambda$ is a measure of the size of the outbreak.

Our preliminary results \cite{bcc21} indicated that it is more effective 
to vaccinate high degree vertices and not neighbors. 
Here, we study more in depth the problem to confirm/infirm these findings.\\
In particular, we ask the following questions : 
which vertex if vaccinated, will reduce most $\lambda$? 
what is the role of the degree ? 
Is it better to vaccinate 2 vertices or 3 vertices instead of 1? 
What role do the eigenvectors of the graph Laplacian play?

To address these questions, we analyze the epidemic matrix $M$.
To estimate the maximum eigenvalue of $M$, we use matrix perturbation theory 
\cite{Kato} where the eigenvalues are written as a power series
of a small parameter.
{The perturbation scheme reveals the interplay
between the topology of the network and the dynamics of the infection.}
We study two different contexts, depending whether the disease
propagates inside a country or between countries.
In the first limit (P1), the disease dynamics and mobility
have the same magnitude.
The corrections at orders 1 and 2 of the
maximal eigenvalue show it is most efficient to vaccinate uniformly the network.
We then examine how $\lambda$ varies when vaccination
is applied along an eigenvector $V^k$ of the Laplacian and
find it is minimum when $k$ is large. 
We illustrate these findings on a seven vertex graph
and give special graphs (complete, stars) for which this argument 
does not hold. 
Finally, we study numerically a more realistic situation 
where the Laplacian has 
weights corresponding to routes more traveled than others and where
again the argument holds.

A second interesting limit (P2) is when the local disease dynamics dominates
the mobility. Then, the eigenvalues depend at 1st order of the
perturbation on the degree of each vertex 
and for the 2nd order on the neighbors. We give an example on
a seven vertex graph. The results confirm that the perturbation
approach gives an excellent approximation of $\lambda$.

The article is organized as follows, section 2 presents the model and the
perturbation method. The limit P1 when the disease dynamics 
and mobility have same magnitude is detailed in section 3 and several
graphs are analyzed numerically in section 4. 
In section 5, we describe the limit P2 when the local dynamics dominates 
the mobility and conclusions are presented in section 6.

\section{The model and the perturbation method}

We recall the model introduced in \cite{bcc21} describing the
propagation of an epidemic on a geographical network
where the vertices are indexed $1,2,\dots,n$
\begin{equation} \label{sil}
\left\{\begin{array}{ll}
  {\dot S} = \alpha L S - \beta S I, \\
  {\dot I} =  \alpha L I  + \beta S I - \gamma I .\\
  {\dot R} =  \alpha L R + \gamma I .\\
 \end{array}\right.
\end{equation}
where $S=(S_1,S_2,\dots,S_n)^T,~~ I=(I_1,I_2,\dots,I_n)^T$
and $R=(R_1,R_2,\dots,R_n)^T$
are respectively the proportions of susceptibles, infected and recovered,
$\beta, \gamma$ are respectively the infection and recovery ratios,
$L$ is the graph Laplacian
matrix \cite{crs01}, and where we denote by $S I$ the vector 
$(S_1 I_1 ,S_2 I_2,\dots,S_n I_n)^T$. 
The quantities $S,I$ and $R$ can be considered as numbers or proportions.
For simplicity, we assume
that the total population at each vertex is the same.

We have the following definition :
\begin{definition}
The graph Laplacian matrix $L$ is the real 
symmetric negative semi-definite matrix, such that \\
$L_{kl} = w_{kl}$ ~~if $k$ and $l$ are connected, 0 otherwise, \\
$L_{kk} =-\sum_{l\neq k} w_{kl},$ \\
where $w_{kl}$ represents the flux between vertices $k$ and $l$.
\end{definition}
We want to understand how the network topology affects the propagation
of the epidemic. Therefore, we assume in most of the article that
the $w$'s are equal to one.

The model (\ref{sil}) is a simplified origin-destination mobility model (like
\cite{bh13}) coupled to an SIR epidemic model since
we assumed symmetry in the transition matrix. 
The diffusion through the graph Laplacian is a first order approximation
of dispersion of all the subjects (susceptibles, infected and recovered)
similar to Fourier's or Ohm's law. 
For example, Murray \cite{murray},
uses such a model in continuum space to describe the propagation of rabies.
\begin{figure}[H]
\centering
\resizebox{8 cm}{4 cm}{\includegraphics[angle=0]{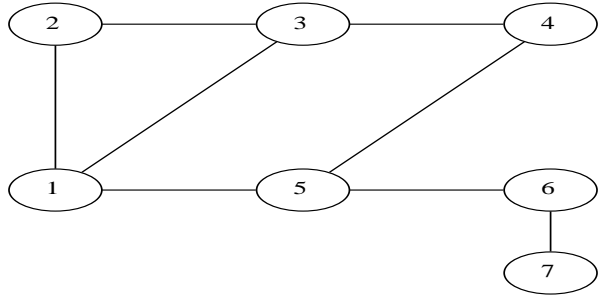}}\hfill
\caption{A seven vertex graph}
\label{graph7}
\end{figure}
To illustrate the model, consider the seven vertex graph 
shown in Fig. \ref{graph7}. The weightless graph Laplacian matrix is
$$
L= \begin{pmatrix}
-3 & 1  & 1 & 0 & 1 & 0 & 0\\
1  & -2 & 1 & 0 & 0 & 0 & 0 \\
1  & 1  & -3& 1 & 0 & 0 & 0 \\
0  & 0  & 1 &-2 & 1 & 0 & 0 \\
1  & 0  & 0 &1  &-3 & 1 & 0 \\
0  & 0  & 0 &0  &1  &-2 & 1 \\
0  & 0  & 0 &0  &0  &1  &-1 
\end{pmatrix}
$$
The graph Laplacian matrix has important properties, see Ref. \cite{crs01}, in 
particular, it is a finite difference approximation of the
continuous Laplacian \cite{Lap2}. The eigenvalues of $L$ are the $n$ 
non-positive real numbers ordered and denoted as follows:
$$ 0 = -\omega_{1}^2 \geq -\omega_{2}^2 \geq ... \geq -\omega_{n}^2.$$
The eigenvectors $\{V^1 , ... , V^n \}$ satisfy
$$L V^j = -\omega_{j}^2  V^j. $$ 
and can be chosen to be orthonormal with respect to the standard 
scalar product in $\mathbb{R}^n$, i.e. 
$(V^i, V^j) = \delta_{i,j}$ where $\delta_{i,j}$ is the Kronecker symbol.
The eigenvector $V^1$ corresponding to $\omega=0$ has equal components.

For small time, observe that equations (\ref{sil}) imply 
$${\dot I}  = M I ,$$
where the epidemic matrix $M$ is defined as.
\begin{definition}
For a graph with Laplacian $L$ and initial proportion of
susceptibles $S$, the epidemic matrix $M$ is
\be\label{epicri} 
M = \alpha L + \beta {\rm diag}(S) -\gamma {\rm Id} , \ee
where ${\rm Id}$ is the identity matrix of order $n$.
\end{definition}
Note that $I$ grows exponentially. The matrix $M$ is symmetric.
Its eigenvalues are real because the eigenvalues of $L$ are real and
the additional terms will shift them on the real
axis.
The maximum eigenvalue $\lambda$ of $M$ gives the initial rate of growth 
of the infected on the network.
We define the epidemic rate in the following way.
\begin{definition}
The epidemic growth rate is the maximum eigenvalue $\lambda$ of $M$.
\end{definition}
Our main goal in this article is to discover the vaccination policy
that minimizes the maximum eigenvalue $\lambda$.
For that we use eigenvalue perturbation theory.

\begin{figure}[H]
\centering
\resizebox{12 cm}{5 cm}{
\includegraphics[angle=90]{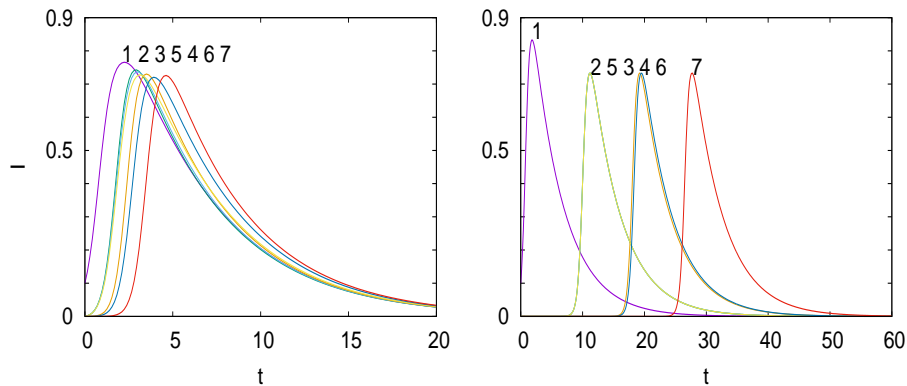}
}\hfill
\caption{
Plot of $I_k(t)$ for the different vertices $k$ for an initial condition
$S_k=1 k=1,..n$, $I_1=0.1$, $\alpha = 0.1$ (left panel) and 
$\alpha = 10^{-10}$ (right panel).
The other parameters are $\beta=2.7, \gamma = 0.2$.
}
\label{slfa}
\end{figure}
Fig. \ref{slfa} shows the time evolution of the infected for the
network shown in Fig.  \ref{graph7}. We plot $I_k(t)$ for 
a fast diffusion (left panel) labeled {\bf P1} and a
slow diffusion (right panel) labeled {\bf P2}. 
\begin{itemize}
\item {\bf P1} is for medium diffusion, like in a highly
connected city or small country. The infected will grow uniformly
across the network. One can then study how a small variation 
$\epsilon s$ of the susceptibles $S$ due to the initial vaccination
affects the epidemic growth rate.  The order 
0 eigenvalues are 
$$ 0 = -\omega_{1}^2 \geq -\omega_{2}^2 \geq ... \geq -\omega_{n}^2.$$
and the order 0 eigenvectors are the eigenvectors of the graph Laplacian
$L$, $\{V^1 , ... , V^n \}$.
 \item {\bf P2} corresponds to small diffusion $\alpha$.
This is a model of long distance travel between vertices like air travel
between different countries. We used it to describe the propagation
of COVID-19 between the main world airports in the spring 2020 \cite{bcc21}.
Here, as shown in Fig. \ref{slfa}, we see a succession of well
separated peaks as the outbreak starts and falls on the different 
vertices.  
We consider here that the perturbation is the graph Laplacian
parameterized by $\alpha$.
The order 0 eigenvalues are ${\rm diag}( \beta S) -\gamma {\rm Id}$
and the order 0 eigenvectors are the ones of the canonical base. \\

\end{itemize}

\subsection{Perturbation theory}

The matrix $M$ can be written as
\be\label{m0pr}
M = M_0 + \epsilon R, 
\ee
where $M_0,R$ depend on the assumptions {\bf P1} or {\bf P2} we make.

The principle of this perturbation theory for eigenvalues and
eigenvectors of a matrix \cite{Kato} is to write expansions of an
eigenvalue $\lambda$ of $M$ and its corresponding eigenvector $v$ as
\begin{eqnarray}
\lambda = \lambda^0 + \epsilon  \lambda^1 + \epsilon^2  \lambda^2+ \epsilon^3  \lambda^3+...  ,\\
v = v_0 + \epsilon  v_1 + \epsilon^2  v_2+ \epsilon^3  v_3+...  
\end{eqnarray}
and write the different orders in $\epsilon$.
For small enough $\epsilon$, this expansion can be 
shown to converge \cite{Kato}.

We introduce the expansions above in the eigenvalue equation  
$M v  = \lambda v,$ and the first three orders in $\epsilon$ yield
\begin{eqnarray}
\label{ord0} (M_0 - \lambda^0)v_0 =0 , \\
\label{ord1} (M_0 - \lambda^0)v_1 = (\lambda^1 -R ) v_0 , \\
\label{ord2} (M_0 - \lambda^0)v_2 = \lambda^2 v_0 + \lambda^1 v_1 -R v_1 .
\end{eqnarray}

These linear equations have solutions if their right-hand side
is orthogonal to the kernel of 
$(M_0 - \lambda^0)^\dagger = (M_0 - \lambda^0)$. This is the
solvability condition.
From the solvability conditions, we obtain $\lambda^1$
and $\lambda^2$ as
\be\label{lambda1}\lambda^1 =  \frac{(v_0, Rv_0) }{(v_0, v_0)},\ee
\be\label{lambda2} \lambda^2 =\frac{(v_0, R v_1 -  \lambda^1 v_1) }{(v_0, v_0)} . \ee
The order 1 eigenvector $v_1$ solves equation (\ref{ord1}).

In both cases {\bf P1} and {\bf P2}, the matrix $M$ is symmetric 
so that it's eigenvalues are real. Then, for
small $\epsilon$, the order of the eigenvalues will not vary and 
the maximum eigenvalue $\lambda$ will reduce to 
the maximum eigenvalue of $M_0$.

The eigenvalues of $M$ are roots of a polynomial whose coefficients
depend on $\epsilon$, therefore they depend analytically on $\epsilon$.
This means that the expansion converges and generically there are no singularities \cite{Kato}.

\section{{\bf P1} Medium diffusion: perturbation results}

In the P1 framework, we assume that $\alpha$ is comparable to $\beta,\gamma$. 
To reduce the dispersion of the disease, the vaccination decreases the 
proportion of susceptibles.
We assume this change to be small, $O(\epsilon)$ so that at a vertex $i$ 
$S_i$ is given by
\be\label{si}
S_i = 1 - \epsilon s_i , \ee
then $- \epsilon s_i$ is the reduction of the number
of susceptibles at vertex $i$ due to vaccination.

The matrix $M$ can be written as
\be\label{P1epicri}
M = M_0 + \epsilon R, ~~
M_0= L + (\beta  -\gamma) {\rm Id}, ~~ R=
- \beta  \begin{pmatrix}
s_1  & 0 & \dots & 0\\
0    & s_2  & \ddots & \vdots\\
\vdots &    \ddots & \ddots & 0\\
0 &  \dots & 0&  s_n
\end{pmatrix}
\ee

In our special case $(M_0 - \lambda^0) = L$ so
the equations (\ref{ord0}, \ref{ord1}, \ref{ord2}) above reduce to
\begin{eqnarray}
\label{ord02} L v_0 =0 , \\
\label{ord12} L v_1 = (\lambda^1 -R ) v_0 , \\
\label{ord22} L v_2 = \lambda^2 v_0 + (\lambda^1 -R ) v_1 .
\end{eqnarray}
The maximum eigenvalue $\lambda$ corresponds to the 0 eigenvalue
of the Laplacian. Then we have
\be\label{lambda02} \lambda^0 = \beta-\gamma . \ee

The Laplacian has real eigenvalues and orthogonal eigenvectors
\be \label{lapla} L V^i = - \omega_i^2 V^i, \ee
where $\omega_1=0$ and $V^1$ is the constant vector.
The other eigenvalues verify
\be\label{order} - \omega_n^2 \le \dots \le  -\omega_1^2=0  . \ee
We assume the graph to be simply connected so that there is only one
eigenvalue zero \cite{crs01}.
The matrix $L$ is therefore singular and
special care must be taken when solving the system.
The standard way to solve the system is to use the singular
value decomposition of $L$. Since $L$ is
symmetric, this reduces to projecting the solution
and the right-hand side
on the eigenvectors of $L$. Therefore, we can choose $v_0=V^1$
where $V^1$ the constant eigenvector is normalized.
The formulas (\ref{lambda1},\ref{lambda2}) become
\be\label{lambda12}\lambda^1 =  (V^1, R V^1),\ee
\be\label{lambda22} \lambda^2 = \left (V^1, (R  -  \lambda^1 \right )v_1)
 , \ee
where $v_1$ solves the linear equation ( \ref{ord12}).

Our main goal is to define possible vaccination policies and reveal 
an optimal one, in the sense that 
$$\lambda = \lambda^0+\epsilon\lambda^1+\epsilon^2\lambda^2+O(\epsilon^3)$$
is minimum. 
Equation (\ref{epicri}) shows that the network determines $\lambda$. 
Following this observation the following questions arise: is it better 
to vaccinate uniformly the network? or if this is not possible, 
what are the best vertices to vaccinate?

In view of these questions, three vaccination strategies are possible
\begin{itemize}
\item [(i)] Reduce  the proportion
of susceptibles of $S_i$ uniformly on all vertices, then
$s_i \ge 0,~~ i=1, \dots , n$.
\item [(ii)] Reduce $S_i$ on some vertices and not others, then
$s_i > 0,$ for some $i=1, \dots , n$. 
\item [(iii)] Adjust $S_i$ globally using an eigenvector $V^k$
of the Laplacian $L$, i.e. $S = V^1 - \epsilon V^k$.
Then $s_i$ can be positive or negative.
\end{itemize}
Approach (i) is to vaccinate uniformly all vertices. This assumes
that we have the logistics to distribute the vaccine throughout the
network. It is the simplest of strategies and will
be the benchmark to test the other strategies. 
Approach (ii) assumes there are a limited number of vertices that can
be vaccinated. Then, we need to choose which ones. \\
Approach (iii) is not practical since we cannot increase $s_i$, we
can only decrease it. Despite this, we consider it in this theoretical section.

We can state the following \\
{\bf Mathematical program }\\
Minimize the epidemic growth rate $\lambda$ such 
that $\sum_{i=1}^n s_i $ is constant \\
For strategy (i) and (ii) $s_i \ge 0$ \\
For strategy (iii) there are no positivity constraints on $s_i$.

\subsection{The first order $\lambda^1$}

First consider the simple situation where all vertices are vaccinated
with the same amount, $s_1=s_2=,\dots,=s_n=1/n$.
Then, 
$$M  = L + (\beta  -\gamma) {\rm Id}  - \beta \epsilon {1\over n} {\rm Id}= 
L + [ \beta (1 -\epsilon {1\over n}) - \gamma] {\rm Id} .  $$
Then the maximum eigenvalue of $M$ is
\be\label{lmin}
\lambda = \beta (1 -\epsilon {1\over n}) - \gamma . \ee
This result is exact.

When the $s_i$ are different, we have the following result.
\begin{theorem}\label{thm0} 
Let $G$ be a connected graph with $n$ vertices
and assume that the number of susceptibles at each vertex $i$ is 
$S_i=1 - \epsilon s_i$. 
Then the epidemic growth rate is
\be\label{lambda12a} 
\lambda = \beta - \gamma -\epsilon {\beta \over n } \sum_{i=1}^n  s_i
+ O(\epsilon^2) . \ee
\end{theorem}

\begin{proof}
This is a direct consequence of equation (\ref{lambda12}) where we chose
$v_0= V^1$ the constant eigenvector.
\end{proof}

We can make the following remarks.  
\begin{itemize}
\item [(i)] When all vertices are vaccinated, we recover the 
result (\ref{lmin}) and the $O(\epsilon^2)$ term is zero.
\item [(ii)] If only one vertex is vaccinated, then
\be\label{lambda12b}
\lambda = \beta - \gamma -\epsilon {\beta \over n }
+ O(\epsilon^2) . \ee
This expression does not depend on the vertex that is vaccinated.
\item [(iii)] Note that $\lambda^1$ is always negative.
\item [(v)] From Theorem \ref{thm0}, $\lambda^1$ is minimal when 
$\sum_{i=1}^n  s_i$ is maximal. In other words,  we can minimize 
$\lambda^1$ and thus the epidemic growth rate $\lambda$ by increasing 
the total percentage of the vaccinated population on the network 
regardless of their location.
\end{itemize}

\subsection{Spectral approach}

Using the eigenvectors of the Laplacian matrix $L$ (\ref{lapla}) one
can calculate $\lambda^2$ in closed form. We have
\begin{theorem}\label{thm2} Let $G$ be a connected graph with $n$ vertices
and assume that the susceptibles at each vertex $i$ are
$S_i=1 - \epsilon s_i$.
Then the epidemic growth rate is
\be\label{lambda2f}
\lambda = \beta - \gamma -\epsilon {\beta \over \sqrt{n} } (s,V^1)
+ \epsilon^2 {\beta \over n} \sum_{k=2}^n { 1 \over \omega_k^2 } (s,V^k)^2
+ O(\epsilon^3) , \ee
where $s=(s_1,s_2,\dots,s_n)^T$ and 
$(V^1,V^2,\dots,V^n)$ are the eigenvectors of the graph Laplacian $L$
corresponding to the eigenvalues
$0=- \omega_1^2 >  -\omega_2^2 \ge \dots \ge -\omega_n^2 $ .
\end{theorem}

\begin{proof}
We have
$$\lambda^2 =(V^1, (R -\lambda^1) v_1) , $$
where
$$L v_1 = (\lambda^1 -R ) V^1  . $$
We expand $v_1$ on the basis of the eigenvectors of $L$
$$v_1 = \sum_{k=2}^n \alpha_k V^k , $$
plug it into the equation above, and then project it onto each
eigenvector $V^k$ to yield $\alpha_k$. We obtain 
\be\label{cv1}
v_1 = \sum_{k=2}^n { 1 \over \omega_k^2 } (R V^1,V^k) V^k ,  \ee
where 
$$ R V^1 = -{\beta \over \sqrt{n} } (s_1, s_2,\dots,s_n)^T$$

We can now compute $\lambda^2$, as 
$$\lambda^2 =(V^1, (R -\lambda^1) v_1) = \sum_{k=2}^n { 1 \over \omega_k^2 } (R V^1,V^k) (V^1,(R -\lambda^1)V^k)$$
We have 
$$(V^1,(R -\lambda^1)V^k) = (V^1,R V^k)$$
because $V^1$ and $V^k$ are orthogonal.
We finally get
$$
\lambda^2 =\sum_{k=2}^n { 1 \over \omega_k^2 } (R V^1,V^k)^2 
={\beta \over n} \sum_{k=2}^n { 1 \over \omega_k^2 } (s,V^k)^2
.$$

\end{proof}

Note that the correction $\lambda^2$ is always positive.

An important theorem follows from the estimate (\ref{lambda2f}).
\begin{theorem}\label{lambda2-minimal} Let $G$ be a connected 
graph with $n$ vertices
and assume that the vector of susceptibles is
$S=1- \epsilon \sqrt{n} V^k$, where $V^k$ is the $k$th eigenvector
of the graph Laplacian $L$.
Then the epidemic growth rate is
\be\label{lambda2a}
\lambda = \beta - \gamma 
+ \epsilon^2 { \beta \over \omega_k^2 } 
+ O(\epsilon^3) . \ee
The eigenvalue $\lambda$ is minimum when $k=n$.
\end{theorem}

\begin{proof}
Choosing   $(s_1,s_2,\dots,s_n)^T = \sqrt{n} V^k$ where $V^k, k\ge 2$
is an eigenvector of $L$ leads to $\lambda^1=0$.   \\
In equation (\ref{lambda12}), we have from the orthogonality of
$V^k$ to $V^1$ 
$$\sum_{i=1}^n s_i= \sum_{i=1}^n V^k_i =0.$$
Assume  $S=\sqrt{n} (V^1 - \epsilon V^k) $. Then 
$$\lambda^2 = \beta { 1 \over \omega_k^2 } .$$
From the order relation (\ref{order}), this quantity decreases
monotonically as $k$ varies from 2 to n.

\end{proof}
In other words, if we choose $(s_1,s_2,\dots,s_n)^T = \sqrt{n}V^k$,
the correction $\lambda^2$ decreases as $k$ increases.

\subsection{Comparing the different vaccination strategies}

From this analysis, we can discuss how different vaccination policies
proposed in the beginning of section 3 affect $\lambda$.
\begin{itemize}
\item [(i)] Uniform vaccination on the network, i.e. 
This corresponds to $ S =\sqrt{n}(1-\epsilon {1\over n})V^1$.
Then 
\be \label{luni} \lambda = \beta (1-\epsilon {1\over n})-\gamma . \ee
This is the optimal vaccination strategy for the network because there
is no positive term $O(\epsilon^2)$.
\item [(ii)]  Vaccination following the eigenvector $V^k$, then 
$S=1 - \epsilon \sqrt{n} V^k$. We then have
$$\lambda^0=\beta-\gamma,~~ \lambda^1=0,~~\lambda^2=
{\beta \over \omega_k^2 }$$
so that 
\be\label{lvec}
\lambda = \beta-\gamma+ \epsilon^2 {\beta \over \omega_k^2 } 
+ O(\epsilon^3).\ee
This expression is minimum for $k=n$. As indicated above, this is not realistic
because we cannot increase $S_j$ at certain vertices $j$.
\item [(iii)] Vaccination of $j < n$ vertices of the network, i.e. 
$S_k = 1 - \epsilon,~~k=1,\dots,j$. \\
This is the most difficult situation because we cannot control 
$\lambda^2$. We have
$$\lambda^0=\beta-\gamma,~~ \lambda^1=- {\beta j \over n}$$
so that
\be\label{lvert}
\lambda = \beta-\gamma-\epsilon {\beta j \over n} +  O(\epsilon^2) .  \ee
Note that when $j=n$, $\lambda^i=0,~~i\ge 2$ and  (ii) reduces to (i). \\
From the results of (ii), a strategy emerges: one can vaccinate 
some vertices $j$ so that the $s$ vector becomes close to an eigenvector
$V^k$, preferably of high order. We will see below that this method
gives a $\lambda$ that is minimal so that this strategy is optimal for $j < n$.

\end{itemize}

\section{{\bf P1}: Two examples} 

Here, we examine numerically the matrix $M$ and its maximal
eigenvalue $\lambda$ for different graphs to emphasize the
role of the graph topology. 
 
\subsection{The 7 vertex graph}

For the 7 vertex graph considered above,
we computed the largest eigenvalue of the matrix $M$ with 
$ \epsilon s = 0.3 V^k$  for $k=2, \dots n$.
We chose $\epsilon = 0.3$ as a small quantity
and $\beta=\gamma=1$ for simplicity.
The results are presented in the table below.
\begin{table}[H]
\begin{center}
\begin{tabular}{|l|c|c|r|}
\hline 
$k$     &  $\lambda$      &  perturbation & relative \\
        &                 &      (\ref{lvec}) & error \\ \hline
2 &$   1.204~10^{-1}$&$  2.314~10^{-1}$&$   4.8~10^{-1}   $\\
3 &$   6.073~10^{-2}$&$  5.958~10^{-2}$&$   1.9~10^{-2}   $\\
4  &$  4.173~10^{-2}$&$  4.086~10^{-2}$&$   2.1~10^{-2}   $\\
5 &$   2.716~10^{-2}$&$  2.791~10^{-2}$&$   2.7~10^{-2}   $\\
6 &$   2.547~10^{-2}$&$  2.405~10^{-2}$&$   5.9~10^{-2}   $\\
7 &$   1.821~10^{-2}$&$  1.825~10^{-2}$&$   2.3~10^{-3}   $\\
\hline
\end{tabular}

\end{center}
\caption{Largest eigenvalue $\lambda$ of $M$ for $s = V^k$ (middle column)
and perturbation estimate (right column).} 
\label{tg7}
\end{table}
Note how  $\lambda$ decreases between $2\leq k \leq 7$.
The  optimal vaccination policy is the one that follows $V^7$.
The eigenvalue $\lambda$ varies from 0.12 to 0.018 as
$s$ follows $V^2$ or $V^{7}$. The perturbation estimate is shown in the
third column and the relative error in the fourth. It is about $\epsilon^2$
except for $k=2$. The eigenvector $V^2$ has large components on vertices 6 and
7 and smaller components on the other vertices. Then, the perturbation 
approach becomes less accurate.

The eigenvectors of the Laplacian are plotted in Fig. \ref{vecg7}.
As expected, the low-order eigenvectors $V^2,~V^3,~V^4$ 
vary on scales comparable
to the size of the graph while the high order eigenvectors
$V^5,~V^6,~V^7$ oscillate on smaller scales.
\begin{figure}[H]
\centering
\resizebox{12 cm}{5 cm}{\includegraphics[angle=0]{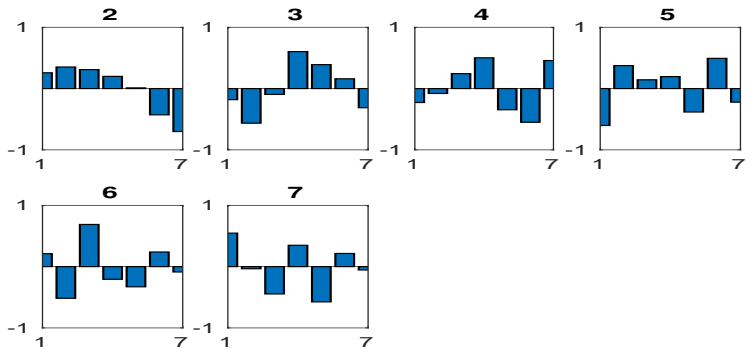}}\hfill
\caption{The eigenvectors $V^i,~i=2,\dots ,7$ of the seven vertex graph
of Fig. \ref{graph7}.}
\label{vecg7}
\end{figure}
It is difficult to relate the results of Table \ref{tg7} to 
the practical situation of vaccinating individual
vertices. To study this, we now vaccinate two or three vertices and 
compute the epidemic growth rate.
The sum of the $S$ vector is the same for both situations, 
and corresponds to a limited amount of vaccines being
distributed over a geographic region. 
We chose the parameters $\epsilon =1.,~~ \beta =1.12 ,~~ \gamma=1$
so that the eigenvalues $\lambda$ are distributed on both sides of zero.
As discussed above, the minimum of $\lambda$ corresponds to 
a uniform vaccination of the network, it is
\be\label{lmin7}\lambda_{\rm min} = \beta (1 -\epsilon {1\over n}) - \gamma= -0.04 .\ee
This quantity will provide a benchmark to measure how efficient the
vaccination is.

For two vaccinated vertices $i,j$, we choose $s_i= 0.5, ~s_j=0.5$
so that $s_i+s_j=1$.
Table \ref{tvg7} gives for $i,j$, the maximum of the projection on the
eigenvectors 
\be\label{proj}
p = {\rm argmax}_{1<k\le n} |(s,V^k)|, \ee
and the epidemic growth rate $\lambda$.
\begin{table}[h]
\begin{center}
\begin{tabular}{|l|c|c|r|}
\hline
i &  j  &   $p$ &  $\lambda$ \\ \hline
 3 &    6 &   6 &  $-2.1751 ~10^{-2}$ \\
 1 &    6 &   4 &  $-1.8658~10^{-2}$ \\
 2 &    6 &   5 & $ -1.6895~10^{-2}$ \\
 3 &    5 &   7 &  $-1.3008~10^{-2}$ \\
 1 &    5 &   5 &  $-1.1325~10^{-2}$ \\
 3 &    7 &   4 &  $-1.0303~10^{-2}$ \\
 4 &    6 &   3 &  $-1.0146~10^{-2}$ \\
\hline
\end{tabular}
\begin{tabular}{|l|c|c|r|}
\hline
i &  j  &   $p$ &  $\lambda$ \\ \hline
 2 &    5 &   6 &  $-9.5849~10^{-3}$ \\
 1 &    7 &   5 & $ -7.9345~10^{-3}$ \\
 1 &    4 &   7 & $ -7.8741~10^{-4}$ \\
 4 &    5 &   3 & $ -5.6481~10^{-3}$ \\
 2 &    7 &   3 & $ -4.7755~10^{-3}$ \\
 4 &    7 &   4 & $ 1.6008~10^{-4}$ \\
 5 &    6 &   4 & $ 2.1799~10^{-4}$ \\
\hline
\end{tabular}
\begin{tabular}{|l|c|c|r|}
\hline
i &  j  &   $p$ &  $\lambda$ \\ \hline
 2 &    4 &   6 &  $6.1549~10^{-3}$ \\
 3 &    4 &   4 &  $7.1014~10^{-3}$ \\
 1 &    3 &   6 &  $7.5348~10^{-3}$ \\
 5 &    7 &   2 &  $8.0539~10^{-3}$ \\
 1 &    2 &   3 &  $1.4141~10^{-2}$ \\
 2 &    3 &   2 &  $1.7100~10^{-2}$ \\
 6 &    7 &   2 &  $4.4036~10^{-2}$ \\
\hline
\end{tabular}
\end{center}
\caption{Maximal eigenvalue $\lambda$ and argmax of the 
projection $p$ (see (\ref{proj}) ) when
vaccinating two vertices $i,j$ for the 7 vertex graph of Fig. \ref{vecg7}.
The parameters are $\epsilon =1.,~~ \beta =1.12 ,~~ \gamma=1$.}
\label{tvg7}
\end{table}
As expected, the largest $\lambda$ corresponds to a $p$
that is maximal on the low order eigenvectors and vice-versa.
This is an average trend and there are some exceptions
such as (1,3) , (2,4). This is because the second-largest
projection is on $V^2$, see Fig. \ref{vecg7}.

For three vaccinated vertices $i,j,k$, we choose $s_i=s_j=s_k=1/3$
so that $s_i+s_j+s_k=1$.
We define the projection similarly to (\ref{proj}). The results are 
presented in Table \ref{tvg72}.
\begin{table}[H]
\begin{center}
\begin{tabular}{|l|c|c|c|r|}
\hline
i &  j  &   k   & $p$ &  $\lambda$ \\ \hline
1 &    4  &  6 &   7 & $  -2.9565~10^{-2}$  \\
 2 &    4  &  6 &   5 & $  -2.9278~10^{-2}$  \\
 1 &    3  &  6 &   6 & $  -2.9002~10^{-2}$  \\
 3 &    5  &  6 &   7 & $  -2.8358~10^{-2}$  \\
 1 &    4  &  7 &   7 & $  -2.7529~10^{-2}$  \\
 2 &    5  &  6 &   4 & $  -2.7500~10^{-2}$  \\
 2 &    4  &  7 &   4 & $  -2.7036~10^{-2}$  \\
 3 &    4  &  6 &   5 & $  -2.6735~10^{-2}$  \\
 1 &    3  &  7 &   6 & $  -2.6709~10^{-2}$  \\
 3 &    5  &  7 &   7 & $  -2.6515~10^{-2}$  \\
 2 &    3  &  6 &   5 & $  -2.5639~10^{-2}$  \\
 2 &    5  &  7 &   6 & $  -2.5531~10^{-2}$  \\
\hline
\end{tabular}
\begin{tabular}{|l|c|c|c|r|}
\hline
i &  j  &   k   & $p$ &  $\lambda$ \\ \hline
 1 &    2  &  6 &   4 & $  -2.5099~10^{-2}$  \\
 1 &    5  &  6 &   4 & $  -2.4357~10^{-2}$  \\
 3 &    4  &  7 &   4 & $  -2.4328~10^{-2}$  \\
 2 &    3  &  7 &   3 & $  -2.2840~10^{-2}$  \\
 1 &    2  &  7 &   3 & $  -2.2372~10^{-2}$  \\
 1 &    5  &  7 &   5 & $  -2.2337~10^{-2}$  \\
 4 &    5  &  6 &   3 & $  -1.8153~10^{-2}$  \\
 4 &    5  &  7 &   3 & $  -1.5852~10^{-2}$  \\
 1 &    4  &  5 &   3 & $  -1.3836~10^{-2}$  \\
 2 &    4  &  5 &   6 & $  -1.2760~10^{-2}$  \\
 1 &    3  &  5 &   5 & $  -1.2451~10^{-2}$  \\
 3 &    4  &  5 &   3 & $  -1.1279~10^{-2}$  \\
\hline
\end{tabular}
\begin{tabular}{|l|c|c|c|r|}
\hline
i &  j  &   k   & $p$ &  $\lambda$ \\ \hline
 1 &    2  &  5 &   4 & $  -9.4617~10^{-3}$  \\
 2 &    3  &  5 &   7 & $  -9.3776~10^{-3}$  \\
 3 &    6  &  7 &   6 & $  -8.8748~10^{-3}$  \\
 2 &    6  &  7 &   2 & $  -7.1971~10^{-3}$  \\
 1 &    6  &  7 &   2 & $  -5.4080~10^{-3}$  \\
 1 &    3  &  4 &   2 & $  -1.0669~10^{-3}$  \\
 1 &    2  &  4 &   7 & $  -2.8709~10^{-4}$  \\
 4 &    6  &  7 &   2 & $  6.4395~10^{-4}$  \\
 2 &    3  &  4 &   2 & $  4.2058~10^{-3}$  \\
 1 &    2  &  3 &   2 & $  9.1003~10^{-3}$  \\
 5 &    6  &  7 &   2 & $  1.1594~10^{-2}$  \\
\hline
\end{tabular}
\end{center}
\caption{Maximal eigenvalue $\lambda$ and projection $p$ when
vaccinating three vertices $i,j,k$ for the 7 vertex graph of Fig. \ref{vecg7}.}
\label{tvg72}
\end{table}
Our results show that vaccinating three vertices $i,j,k$ gives 
eigenvalues that, on average, are minimal when the projection $|(s,V^p)|$ 
corresponds to a large $p$.
Of course, the trend is general and there are a few exceptions.
\begin{figure}[H]
\centering
\resizebox{12 cm}{5 cm}{
\includegraphics[angle=90]{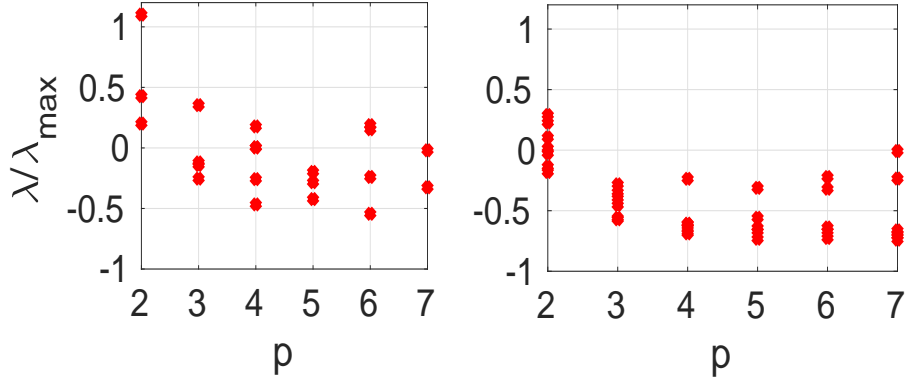}
}\hfill
\caption{Ratio $\lambda/\lambda_{\rm max}$ of the epidemic matrix
as a function of the projection (\ref{proj}) for
two (left) and three (right) vaccinated vertices.
}
\label{r7v23}
\end{figure}
The results shown in the tables (\ref{tvg7},\ref{tvg72}) are summarized
in Fig. \ref{r7v23} where we present 
$\lambda/\lambda_{\rm max}$ vs. the maximal projection (\ref{proj}) 
for two (left panel) and three (right panel) vaccinated vertices.
The normalization factor is
$$\lambda_{\rm max} = |\lambda_{\rm min}|= 4~10^{-2} . $$

As expected from the perturbation theory, 
$\lambda/\lambda_{\rm max}$ decreases on average when $p$ 
increases. Exceptions occur when the second largest projection is
on a low order eigenvector. For example, vaccinating vertices
$i=1,~j=3$ (left panel) leads to $\lambda/\lambda_{\rm max}=0.2$. The 
projections of the vector $s$ onto the 
$V^k, k=2,\dots, 7$ are 
$$(0.2836,~0.1381,~0.0081,~0.2295,~0.4498,~0.0504) . $$
The projection is largest for $k=6$ and the second largest value is for $k=2$.
If instead, vertices $i=2,~j=5$ are vaccinated, we get
$\lambda/\lambda_{\rm max} = -0.55$. The projection vector in that case
is
$$(0.1814,~0.0869,~0.2127,~0.0038,~0.4226,~0.3047)$$ 
whose components are largest for $k=6,7,3,..$ in that order.
The theorem (\ref{lambda2-minimal})  explains the difference 
in $\lambda$ observed for the two situations.

In Fig. \ref{r7v23}, $\lambda/\lambda_{\rm max}=-1$ is the
limit corresponding to a uniform vaccination of the network. We can
then compare how vaccinating two or three vertices changes $\lambda$. 
Fig. \ref{r7v23} shows that in average, it is better to vaccinate 
three vertices rather than two because the values are closer to
the limit $\lambda/\lambda_{\rm max}=-1$. The spread in the values of 
$\lambda/\lambda_{\rm max}$ is also reduced on average for three vaccinated vertices
as opposed to two.

\subsection{Special graphs: complete graphs and stars}


There are classes of graphs for which choosing $s=R V^1 = \sqrt{n} V^k$, 
with $k$ large, does not necessarily affect $\lambda^2$. For these
graphs, the eigenvalues $-\omega_k^2$ are equal so that
the ratio
$${ 1 \over \omega_k^2 } (s,V^k)^2$$
in the sum (\ref{lambda2f}) does not decrease as $k$ increases.

One example is the class of complete graphs $K_n$.
\begin{definition}[Complete graph $K_n$]
A clique or complete graph $K_n$ is a graph where
every pair of distinct vertices is connected by a unique edge.
\end{definition}
The clique $K_n$ has
eigenvalue $-n$ with multiplicity $n-1$  and eigenvalue $0$.
The eigenvectors for eigenvalue $n$  can be
chosen as $v^k = e^1 -e^k,~~ k=2,\dots,n$.
Table \ref{tcomplete} shows the eigenvalue $\lambda$ of $M$
for $K_4,K_5$, and $K_7$ from left to right. As expected
there are no significant changes in $\lambda$ as a function of $k$. 
\begin{table}[H]
\centering
\begin{tabular}[t]{|l|r|} \hline
$k$     &  $\lambda$     \\ \hline
2  &      0.021156\\
3  &       0.023520\\
4  &        0.021153 \\ \hline
\end{tabular}
\begin{tabular}[t]{|l|r|}
\hline
$k$     &  $\lambda$     \\ \hline
2    &        0.019695 \\
3    &        0.019117 \\
4    &        0.019117 \\
5    &     0.019117 \\
\hline
\end{tabular}
\begin{tabular}[t]{|l|r|}
\hline
$k$     &  $\lambda$     \\ \hline
2   &      0.013034 \\
3   &      0.013184  \\
4   &      0.013594  \\
5   &     0.011805 \\
6   &      0.012404  \\
7   &      0.012313 \\ \hline
\end{tabular}
\caption{Largest eigenvalue $\lambda$  of $M$ for $s = 0.3 \sqrt{n} V^k$
for complete graphs $K_4,K_5$ and $K_7$ from left to right.}
\label{tcomplete}
\end{table}

Another special class of graphs where many eigenvalues are equal are stars.
For these a single vertex, say $1$, is connected to the $n-1$ other vertices.
The eigenvalues with their multiplicities denoted as exponents are
$$ 0^1,~~ (-1)^{n-2},~~ ,(-n)^1  .$$
Eigenvectors for $-1$ can be chosen as $e^{2}-e^i~~(i=3,\dots,n)$.  \\
The eigenvector for $-n$ is $(1,-1/(n-1),\dots,-1/(n-1))^T$. \\
Table \ref{tstar} shows $\lambda$ as a function of $k$ for $S_5,S_7$
and $S_{10}$. 
\begin{table}[H]
\centering
\begin{tabular}[t]{|l|r|} \hline
$k$     &  $\lambda$     \\ \hline
2   &    0.10643  \\
3   &    0.075781  \\
4   &    0.12306  \\
5   &    0.019695  \\ \hline
\end{tabular}
\begin{tabular}[t]{|l|r|}
\hline
$k$     &  $\lambda$     \\ \hline
2   &    0.10974 \\
3   &   0.099318  \\
4   &  0.10453 \\
5   &      0.14166 \\
6   &      0.12715 \\
7   &   0.014059 \\ \hline
\end{tabular}
\begin{tabular}[t]{|l|r|}
\hline
$k$     &  $\lambda$     \\ \hline
2     &    0.049159 \\
3     &   0.055638 \\
4     &   0.19283 \\
5     &    0.19283 \\
6     &    0.19283 \\
7     &    0.19283  \\
8     &  0.19283 \\
9     &    0.19283\\
10    &  0.0097722\\ \hline
\end{tabular}
\caption{Largest eigenvalue $\lambda$  of $M$ for $s = 0.3 \sqrt{n} V^k$
for star graphs $n=5,7$ and $10$ from left to right}
\label{tstar}
\end{table}

\subsection{A more realistic case: France}

The practical case of vaccinating a whole country can be tackled
using our approach. Fig. \ref{graph-rail} 
shows a map of the main railway lines in France. 
The fast lines are presented as continuous edges while the
slower ones are dashed. Because of these different mobilities,
we need to introduce weights in the graph Laplacian. We chose
weights 1 and 0.5 for respectively the fast and slow 
connections. We computed the epidemic growth rate $\lambda_k$ 
for $s = 0.3 \sqrt{n} V^k, 
~~~k=2, \dots n$ for the parameters $\beta=2, \gamma=1$.
\begin{figure}[H]
\centering
\includegraphics[scale=0.37]{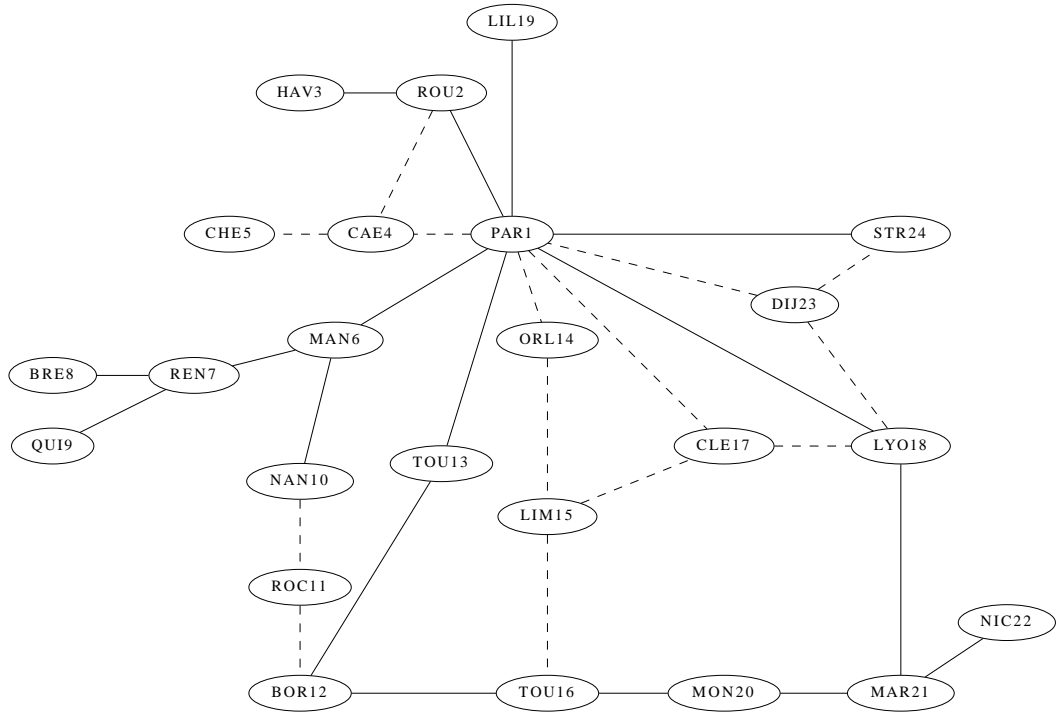}
 \caption{The main railway network in France where the fast (resp. slow)
lines are shown as continuous (resp. dashed) edges.}
\label{graph-rail}
\end{figure}
The results are shown in 
Fig. \ref{lofk01051} shows $\lambda_k$ vs. $k$ for the unweighted and 
weighted graphs respectively on the left and right panels.

\begin{figure}[H]
\centering
\resizebox{15 cm}{5 cm}{
\includegraphics[angle=90]{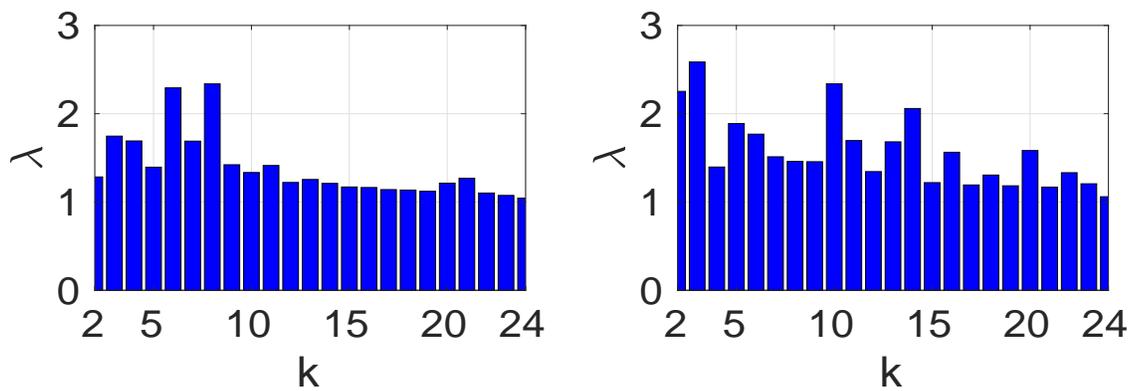}
}\hfill
\caption{Largest eigenvalues of $M$ for $s = 0.3 \sqrt{n} V^k$ for the
graph of France with weights 1 (continuous links) and 1 and 0.5 
respectively from left to right. 
The parameters are $\beta=2, \gamma=1$.}
\label{lofk01051}
\end{figure}
For uniform weights, the topology of the graph controls $\lambda_k$
so that, as expected, it decays as $k$ increases.
There is a factor of 2.3 between the largest and the smallest $\lambda_k$. \\

For the weighted graph (right panel), large eigenvalues occur up to $k=15$.
This graph is close to a star with center (PAR1) of highest degree 
(10); we then expect results to be close to the ones for stars. 
Note also the anomalous large $\lambda$
occuring for $V^{10}$ and $V^{14}$. These eigenvectors are localized
on vertices 8, 9 and 4,5 respectively. Despite this, the trend remains
that $\lambda$ is smallest for large $k$.


We conclude this section by examining how vaccinating two vertices
$i,j$ affects $\lambda$. As above, we chose $s_i=s_j=0.5$.
The reference value $\lambda_{\rm min}$ from (\ref{lmin7}) corresponds
to a uniform vaccination, a reduction of $s$ by 1 distributed 
uniformly over the network. We have
\be\label{lmin8}\lambda_{\rm min} = 
\beta (1 -\epsilon {1\over n}) - \gamma= -0.0129 ,\ee
where we chose $\beta=1.03$ and $\gamma=1$. We choose
$\lambda_{\rm max}=|\lambda_{\rm min}|$

The eigenvalue ratio $\lambda / \lambda_{\rm max}$ vs. the projection
argument $p$ from (\ref{proj}) is shown in Fig. \ref{v2rail}
for the weighted (left panel) and unweighted graphs (right panel).
As expected, the ratio decreases as $p$ increases.
We can reduce the rate of infections significantly so that
in some cases it becomes negative and the outbreak is suppressed.
Interestingly, the weights do not affect this general qualitative
result.
\begin{figure}[H]
\centering
\resizebox{12 cm}{5 cm}{
\includegraphics[angle=90]{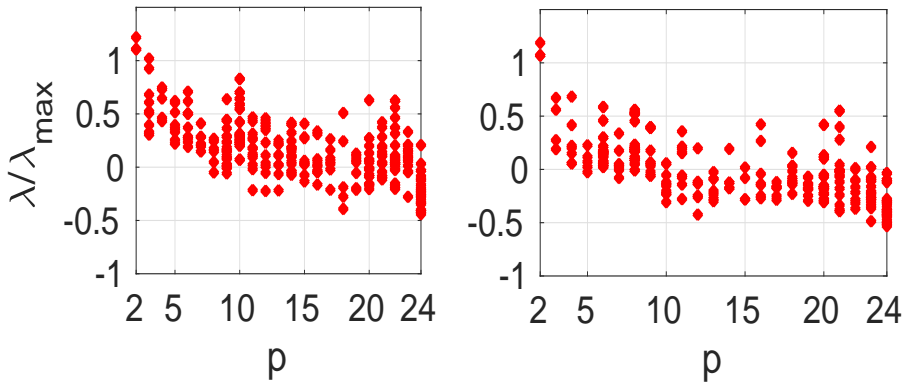}
}
\hfill
\caption{Ratio $\lambda/\lambda_{\rm max}$ of the epidemic matrix
as a function of the projection $p$ (\ref{proj}) for
two vaccinated vertices and weighted (left) or unweighted (right) Laplacian.
The parameters are $\beta= 1.03$ and $\gamma= 1$.}
\label{v2rail}
\end{figure}

\section{{\bf P2} Small diffusion }

In this section, we consider that the diffusion is small compared to
the $S I$ term. This was the situation in 2020 when the
first wave of the COVID19 epidemic hit successively different countries
after starting in China.

\subsection{Perturbation theory}

The epidemic matrix $M$ can be written as
\be\label{P2epicri}
M = M_0 + \alpha L ; ~~ M_0= 
 \beta  \begin{pmatrix}
S_1  & 0 & \dots & 0\\
0    & S_2  & \ddots & \vdots\\
\vdots &    \ddots & \ddots & 0\\
0 &  \dots & 0&  S_n
\end{pmatrix} - \gamma {\rm Id}
\ee
The perturbation 
$R$ of the previous section is $L$ and the perturbation parameter
is $\alpha$.

The eigenvalues and eigenvectors of $M_0$ are
\be\label{P2eig}
\lambda_k^0 =\beta S_k -\gamma,~~e^k,~~ k=1,2,\dots n\ee
where $e^k$ is the canonical vector of order $k$.

We have the following theorem
\begin{theorem}\label{thm10} Let $G$ be a connected graph with $n$ vertices,
assume that the number of susceptibles at each vertex $i$ is $S_i$ with
a vertex $k$ such that $S_k$ is the only maximum and $S_k -S_j > C \alpha$
for all neighbors $j$ of $k$ ($j\sim k$) where $C$ is a constant.
Then the epidemic growth rate is given by
\be\label{lambda11}
\lambda_k = \beta S_k - \gamma -\alpha d_k 
- {\alpha^2 \over \beta} \sum_{j\sim k} {1 \over S_k -S_j   }
+ O(\alpha^3) , \ee
for $k=1,\dots,n$,   where $d_k$ is the degree of vertex $k$.
\end{theorem}

\begin{proof}
Consider eigenvalue $\lambda^0=\lambda_k$, its eigenvector is
$v_0^k=e^k$, in the following we omit the index $k$ unless it is
absolutely necessary.
Then from (\ref{lambda1}), we have
$$\lambda_k^1 =  \frac{(v_0, L v_0) }{(v_0, v_0)} = (e^k, L e^k)= -d_k.$$
From (\ref{ord1}) the eigenvector $v_1^k$ solves 
$$ (M_0 - \lambda^0)v_1^k = (\lambda^1 -L ) v_0 =(\lambda^1 -L ) e^k . $$

Assume $k=1$ for simplicity. Then, the equation above can be rewritten as 
$$
\begin{pmatrix}
0  & 0 & \dots & 0\\
0    & S_2 -\gamma-\lambda^1 & \ddots & \vdots\\
\vdots &    \ddots & \ddots & 0\\
0 &  \dots & 0&  S_n -\gamma-\lambda^1 
\end{pmatrix} v_1 = 
\begin{pmatrix} -d_1 \\ 0 \\ \ddots \\ 0\end{pmatrix}
- L e^1= 
\begin{pmatrix} 0 \\-L_{2 1}  \\ \ddots \\ -L_{n 1}\end{pmatrix}$$
Then for general $k$
$$v_1^k = 
({L_{1 k} \over \beta (S_1 -S_k)} , {L_{2 k} \over \beta (S_2 -S_k)} , \dots ,0, \dots , {L_{n k} \over \beta (S_n -S_k)} )^T$$
where the $0$ is at position $k$.

From $v_1^k$ one computes $\lambda^2_k$ using (\ref{lambda2}),
$$\lambda^2_k = (e^k, L v_1^k -  \lambda^1 v_1^k) . $$
The vectors $e^k$ and $v_1^k$ are orthogonal so that
$$\lambda^2_k = (e^k, L v_1^k).$$
Assume again $k=1$ for clarity, this inner product is
$$(e^1, L v_1^1)= \sum_{i=1}^n e^1_i \sum_{j=1}^n L_{i j} {v_1^1}_j=
{L_{1 2}^2 \over \beta (S_2 -S_1)}
+{L_{1 3}^2 \over \beta (S_3 -S_1)}
+\dots
+{L_{1 n}^2 \over \beta (S_n -S_1)}$$
$$= \sum_{j\sim 1} {L_{1 j}^2 \over \beta (S_j -S_1)}= {1\over \beta} \sum_{j\sim 1} {1\over S_j -S_1} $$
because $L_{1 j}=1$ if $j\sim 1$ and zero
otherwise.
It is then easy to get $\lambda^2_k$ for general $k$.

\end{proof}

We can make the following comments.
\begin{itemize}
\item The first order in the perturbation comes from the degree.
Then for equal $S_k$, a vertex with a larger degree will yield a smaller $\lambda$.
\item The influence of neighbors appears at second order with 
the term $S_k -S_j$ in the denominator. When $S_k -S_j$ is small, this term
can be large and the approximation breaks down. One should then use the approach P1.
\end{itemize}

\subsection{Illustration on the graph with 7 vertices}
To illustrate these calculations, we consider 
the seven vertex graph shown in Fig. \ref{graph7}.
We chose $\alpha =0.3,~ \beta =9 , \gamma=6$ and
$$ S=(0.5~~ 0.9~~ 0.5~~ 0.5~~ 0.5~~ 0.5~~ 0.5 )^T .$$
Table \ref{p1tg7} shows the epidemic growth rate 
when $S_k=0.9$ together with its
degree and the first and second orders of perturbation.
We change $S_k$ following a circular permutation.
The order 0 is $$\beta S_k -\gamma \approx 2.1$$
\begin{table}[H]
\begin{tabular}{|l|c|c|c|c|r|}
\hline
&    &        &     &       \\ 
vertex $k$ &  degree& $\lambda^0+\alpha \lambda^1$ & $\lambda^0+\alpha \lambda^1+ \alpha^2 \lambda^2$ & exact $\lambda$ \\
 &   &     &       &       \\ \hline
1 & 3 & 1.2 & 1.275 & 1.28 \\
2 & 2 & 1.5 & 1.55  & 1.55 \\
3 & 3 & 1.2 & 1.275 & 1.28 \\
4 & 2 & 1.5 & 1.55  & 1.55 \\
5 & 3 & 1.2 & 1.275 & 1.28 \\
6 & 2 & 1.5 & 1.55  & 1.55 \\
7 & 1 & 1.8 & 1.825 & 1.823 \\
\hline
\end{tabular}
\caption{Epidemic growth rate $\lambda$ when 
$S_k=0.9$ for $k$ shown on first column. The second, third, fourth and fifth
columns are respectively the degree, first order, second order and exact calculations. }
\label{p1tg7}
\end{table}
One sees that the perturbation approach is in excellent 
agreement with the exact calculations.
As the degree decreases from 3 to 1, the eigenvalue increases from 1.28
to 1.82.

To see the effect of vaccination, we reduced the maximum $S_k$ from
0.9 to 0.8 and recomputed $\lambda$. The order 0 is 
$$\beta S_k -\gamma \approx 1.2 .$$
The results are shown in Table \ref{p1ag7}.
\begin{table}[H]
\begin{tabular}{|l|c|c|c|c|c|r|}
\hline
&    &        &     &       & \\
vertex $k$ &  degree& $\lambda^0+\alpha \lambda^1$ & $\lambda^0+\alpha \lambda^1+ \alpha^2 \lambda^2$ & exact $\lambda$  & ${\Delta \lambda \over \lambda}$\\
&    &        &     &       &   \\ \hline
1  & 3  & 0.3 & 0.4  & 0.41 &  0.68 \\
3  & 3  & 0.3 & 0.4  & 0.41&  0.68 \\
5  & 3  & 0.3 & 0.4  & 0.41&  0.68 \\
2  & 2  & 0.6 & 0.67 & 0.67 & 0.55  \\
4  & 2  & 0.6 & 0.67 & 0.66&  0.55 \\
6  & 2  & 0.6 & 0.67 & 0.67&  0.55 \\
7  & 1  & 0.9 & 0.93 & 0.93&  0.50 \\
\hline
\end{tabular}
\caption{Same as Fig. \ref{p1tg7}  except $S_k=0.8$. The rows have been 
sorted to have $\lambda$ increasing.}
\label{p1ag7}
\end{table}
Note again the excellent agreement between the perturbation
theory and the exact calculation. For this simple situation,
$\lambda$ is fixed by the degree of the vertex. The rows have been
sorted out in descending degree and one sees that $\lambda$ increases
as the degree decreases. To compare the situations $S_{max}=0.9$
and $S_{max}=0.8$ we plot in the last column 
$${\Delta \lambda \over \lambda}= 
{\lambda_{S=0.9}-\lambda_{S=0.8} \over \lambda_{S=0.9}} .$$ 
This quantity decreases from $0.68$ for degree 3
to 0.5 for degree 1. It is therefore best to vaccinate high 
degree vertices.

When the distribution of $S$ is not uniform, for the same degree,
$\lambda$ will depend on the neighbors $j$ of $k$ through the term
$S_k-S_j$. The epidemic growth rate will be smaller
for vertices $k$ where $S_k -S_j$ small and larger otherwise.

\section{Conclusion}

We studied a vaccination policy on a simple model of
an epidemic on a geographical network by examining the 
maximum eigenvalue $\lambda$ of a matrix formed by the
susceptibles and the graph Laplacian matrix $L$: the epidemic growth rate. 
For that, we analyzed the epidemic matrix $M$ using
perturbation theory.

When the mobility and the disease dynamics are of the same order
(case P1), the zero order eigenvectors are the ones of the graph
Laplacian matrix $L$. The epidemic grows uniformly on the network.
We found that it is most effective to "vaccinate" the 
network uniformly. If this is not possible, then it is
best to vaccinate two or more vertices that follow the 
eigenvector of $L$ of highest order $V^n$.

To design strategies to reduce the spread of the infection, we 
illustrate these results with two examples, the first on a general 
graph of 7 vertices and eight edges. 
For the P1 case, we show that, in average it is better 
to vaccinate three vertices rather than two. These vertices should
have a maximal projection on a high order eigenvector.
The second example is a graph of the main railway lines of France, it is a 
weighted graph of 24 vertices and 31 edges of two types, fast and slow lines. 
We compute the maximal eigenvalue and show that vaccinating vertices
along this eigenvector reduces the growth rate of the infection 
compared to vaccinating along a low order eigenvector.

The second limit P2 is when the local disease dynamics
dominates the mobility, the order zero eigenvectors 
are the ones of the canonical basis and the epidemic 
grows locally on the vertex $k$ where $S_k$ is maximum.
The perturbation theory indicates
that it is most effective to vaccinate high degree vertices and 
that it is not efficient to vaccinate neighboring vertices.

These results answer the questions of the introduction as to the 
role of the degree, which vertex or vertices to vaccinate, the
influence of the eigenvectors of the Laplacian, etc.. \\
A future study could be the analysis of a non Markovian model
of mobility like \cite{satdietz95} or \cite{colizza08} 
to see how the above results are modified. In practice, we would
need to estimate the parameters from data. Also the matrix involved
will be much more complicated.

\section*{Acknowledgment}
The authors thank the region of Normandy and the European
Union for support through the grant
Mod\'elisation et Analyse des Syst\`emes Complexes en Biologie (MASyComB).

\end{document}